\documentclass[10pt]{article}

\usepackage{lmodern}
\usepackage[T1]{fontenc}

\usepackage{latexsym}
\usepackage{amssymb}
\usepackage{amsthm}
\usepackage{amsmath}

\usepackage{amstext}
\usepackage{amsopn}

\usepackage{enumerate}
\usepackage{graphics}
\usepackage{tikz}
\usepackage[all]{xy}

\newtheorem{theorem}{Theorem}

\newtheorem{lemma}[theorem]{Lemma}
\newtheorem {example}[theorem]{Example}

\newcommand{\sgraph}{G}

\newcommand{\weight}{w}
\newcommand{\neighbour}{N}
\newcommand{\products}{\mathcal{P}}

\newcommand{\snet}{\mathcal{S}}
\newcommand{\prodset}{P}

\newcommand{\obar}[1]{\overline{#1}}

\newcommand{\srcnodes}{\mathit{source}}
\newcommand{\payoff}{p}
\newcommand{\strprofile}{s}

\newcommand{\agents}{\mathcal{A}}

\newcommand{\nat}{\mathbb{N}}

\newcommand{\inflset}{\mathcal{N}}

\newcommand{\constutil}{c_0}

\newcommand{\bfe}[1]{\begin{bfseries}\emph{#1}\end{bfseries}\index{#1}}
\newcommand{\myra}{\mbox{$\:\rightarrow\:$}}
\newcommand{\fa}{\mbox{$\forall$}}

\newcommand{\LL}{\mbox{$\ldots$}}
\newcommand{\sse}{\mbox{$\:\subseteq\:$}}
\newcommand{\ES}{\emptyset}
\newcommand{\NI}{\noindent}

\newcommand{\II}{\vspace{2 mm}}
\newcommand{\HB}{\hfill{$\Box$}}
\newcommand{\VV}{\vspace{5 mm}}
\newcommand{\oldbfe}[1]{\begin{bfseries}\emph{#1}\end{bfseries}}

\title{Choosing Products in Social Networks}

\title{Choosing Products in Social Networks}

\author{Sunil Simon
    \thanks{%
      Centre for Mathematics and Computer Science (CWI)
  }
  \and 
Krzysztof R. Apt
    \thanks{%
      Centre for Mathematics and Computer Science (CWI)
      and ILLC, University of Amsterdam, The Netherlands
    }
}

\date{}

\begin{document}

\maketitle

\begin{abstract}
We study the consequences of adopting products by agents who form a
social network. To this end we use the threshold model introduced in
\cite{AM11}, in which the nodes influenced by their neighbours can
adopt one out of several alternatives, and associate with each such
social network a strategic game between the agents. The possibility of
not choosing any product results in two special types of (pure) Nash
equilibria.
  
We show that such games may have no Nash equilibrium and that
determining the existence of a Nash equilibrium, also of a special
type, is NP-complete. The situation changes when the underlying graph
of the social network is a DAG, a simple cycle, or has no source
nodes. For these three classes we determine the complexity of
establishing whether a (special type of) Nash equilibrium exists.
  
We also clarify for these categories of games the status and the
complexity of the finite improvement property (FIP). Further, we
introduce a new property of the uniform FIP which is satisfied when
the underlying graph is a simple cycle, but determining it is
co-NP-hard in the general case and also when the underlying graph has
no source nodes. The latter complexity results also hold for verifying
the property of being a weakly acyclic game.
\end{abstract}

\section{Introduction}

\subsection{Background}

Social networks are a thriving interdisciplinary research area with
links to sociology, economics, epidemiology, computer
science, and mathematics.  A flurry of numerous articles and recent
books, see, e.g., \cite{EK10}, testifies to the
relevance of this field. It deals with such diverse topics as
epidemics, analysis of the connectivity, 
spread of certain patterns of social behaviour, effects of
advertising, and emergence of `bubbles' in financial markets.

One of the prevalent types of models of social networks are the {\em
  threshold models} introduced in~\cite{Gra78}.  In such a setup each
node $i$ has a threshold $\theta(i) \in (0,1]$ and adopts an `item'
  given in advance (which can be a disease, trend, or a specific
  product) when the total weight of incoming edges from the nodes that
  have already adopted this item exceeds the threshold.  One of the
  most important issues studied in the threshold models has been that
  of the spread of an item, see, e.g., \cite{Mor00,KKT03,Che09}.  From
  now on we shall refer to an `item' that is spread by a more specific
  name of a `product'.

In this context very few papers dealt with more than one product.  One
of them is~\cite{IKMW07} with its focus on the notions of
compatibility and bilinguality that result when one adopts both
available products at an extra cost.  Another one is~\cite{BFO10},
where the authors investigate whether the algorithmic approach
of~\cite{KKT03} can be extended to the case of two products.

In \cite{AM11} the authors introduced a new threshold model of a
social network in which nodes (agents) influenced by their neighbours
can adopt one out of \emph{several} products. This model allowed us to
study various aspects of the spread of a given product through a
social network, in the presence of other products.  We analysed from
the complexity point of view the problems of determining whether
adoption of a given product by the whole network is possible
(respectively, necessary), and when a unique outcome of the adoption
process is guaranteed.  We also clarified for social networks without
unique outcomes the complexity of determining whether a given node has
to adopt some (respectively, a given) product in some (respectively,
all) final network(s), and the complexity of computing the minimum and
the maximum possible spread of a given product.

\subsection{Motivation}
\label{subsec:motiv}

Our interest here is in understanding and predicting the behaviour of
the consumers (agents) who form a social network and are confronted
with several alternatives (products).  To carry out such an analysis
we use the above model of~\cite{AM11} and associate with each such
social network a natural strategic game. In this game the strategies
of an agent are products he can choose. Additionally a `null' strategy
is available that models the decision of not choosing any product.
The idea is that after each agent chose a product, or decided not to
choose any, the agents assess the optimality of their choices
comparing them to the choices made by their neighbours.  This leads to
a natural study of (pure) Nash equilibria, in particular of those in
which some, respectively all, constituent strategies are non-null.

Social network games are related to graphical games of \cite{KLS01},
in which the payoff function of each player depends only on a (usually
small) number of other players.  In this work the focus was mainly on
finding mixed (approximate) Nash equilibria.  However, in graphical
games the underlying structures are undirected graphs.  Also, social
network games exhibit the following \bfe{join the crowd} property: the
payoff of each player depends only on his strategy and on the set of
players who chose his strategy and weakly increases when more players
choose his strategy.

Since these games are related to social networks, some natural
special cases are of interest: when the underlying graph is a DAG, has
no source nodes or a simple cycle which is a special case of a graph
without source nodes.
Such social networks correspond respectively to a hierarchical
organization or to a `circle of friends', in which everybody has a
friend (a neighbour). Studying Nash equilibria of these games and
various properties defined in terms of improvement paths allows us to
gain better insights into the consequences of adopting products.

\subsection{Related work}

There are a number of papers that focus on games associated with
various forms of networks, see, e.g., \cite{TW07} for an overview.  A more
recent example is~\cite{AFPT10} that analyses a strategic game between
players being firms who select nodes in an undirected graph in order
to advertise competing products via `viral marketing'.  However, in
spite of the focus on similar questions concerning the existence and
structure of Nash equilibria and on their reachability, from a
technical point of view, the games studied here seem to be unrelated to
the games studied elsewhere.

Still, it is useful to mention the following phenomenon.  When the
underlying graph of a social network has no source nodes, the game
always has a trivial Nash equilibrium in which no agent chooses a
product. A similar phenomenon has been recently observed
in~\cite{BK11} in the case of their network formation games, where
such equilibria are called degenerate.  Further, note that the `join
the crowd' property is exactly the opposite of the defining property
of the congestion games with player-specific payoff functions
introduced in~\cite{Mil96}. In these game the payoff of each player
weakly decreases when more players choose his strategy.  Because in
our case (in contrast to~\cite{Mil96}) the players can have different
strategy sets, the resulting games are not coordination games.

\section{Preliminaries}
\label{sec:prelim}

\subsection{Strategic games}

Assume a set $\{1, \ldots, n\}$ of players, where $n > 1$.  A
\bfe{strategic game} for $n$ players, written as $(S_1, \ldots, S_n,
p_1, \ldots, p_n)$, consists of a non-empty set $S_i$ of
\bfe{strategies} and a \bfe{payoff function} $p_i : S_1 \times \ldots
\times S_n \myra \mathbb{R}$,
for each player $i$.

Fix a strategic game
$
G := (S_1, \ldots, S_n, p_1, \ldots, p_n).
$
We denote $S_1 \times \cdots \times S_n$ by $S$, 
call each element $s \in S$
% \bfe{joint strategy},or 
a \bfe{joint strategy},
denote the $i$th element of $s$ by $s_i$, and abbreviate the sequence
$(s_{j})_{j \neq i}$ to $s_{-i}$. We also write $(s_i,
s_{-i})$ instead of $s$.  
% Finally, we abbreviate $\times_{j \neq i}
% S_j$ to $S_{-i}$.
We call a strategy $s_i$ of player $i$ a \bfe{best response} to a 
joint strategy $s_{-i}$ of his opponents if
$
\fa s'_i \in S_i \ p_i(s_i, s_{-i}) \geq p_i(s'_i, s_{-i}).
$
Next, we call a joint strategy $s$ a \bfe{Nash equilibrium} if each
$s_i$ is a best response to $s_{-i}$, that is, if
\[
\fa i \in \{1, \ldots, n\} \ \fa s'_i \in S_i \ p_i(s_i, s_{-i}) \geq p_i(s'_i, s_{-i}).
\]

Given a joint strategy $s$ we call the sum $\mathit{SW}(s)=\sum_{j =
  1}^{n} p_j(s)$ the \bfe{social welfare} of $s$.  When the social
welfare of $s$ is maximal we call $s$ a \bfe{social optimum}.  Recall
that, given a finite game that has a Nash equilibrium, its \bfe{price
  of anarchy} (respectively, \bfe{price of stability}) is the ratio
$\frac{\mathit{SW}(s)}{\mathit{SW}(s')}$ where $s$ is a social optimum
and $s'$ is a Nash equilibrium with the lowest (respectively, highest)
social welfare. For division by zero, we interpret the
outcome as $\infty$.

Next, we call a strategy $s_i$ of player $i$ a \bfe{better response}
given a joint strategy $s$ if $ p_i(s'_i, s_{-i}) > p_i(s_i, s_{-i})$.
Following the terminology of \cite{MS96}, a \bfe{path} in $S$ is a
sequence $(s^1, s^2, \LL)$ of joint strategies such that for every $k
> 1$ there is a player $i$ such that $s^k = (s'_i, s^{k-1}_{-i})$ for
some $s'_i \neq s^{k-1}_{i}$.  A path $\xi$ is called an
\bfe{improvement path} if it is maximal and for all $k$ smaller than
the length of $\xi$, $p_i(s^k) > p_i(s^{k-1})$, where $i$ is the
player who deviated from $s^{k-1}$.  The last condition simply means
that each deviating player selects a better response.  A game has the
\bfe{finite improvement property} (in short, \bfe{FIP}) if every
improvement path is finite. Obviously, if a game has the FIP, then it
has a Nash equilibrium--the last element of each path.

Finally, recall that a game is called \bfe{weakly acyclic} (see
\cite{Mil96}) if for every joint strategy there exists a finite
improvement path that starts at it.

\subsection{Social networks}

We are interested in specific strategic games defined over social
networks. In what follows we focus on a model of social networks
recently introduced in \cite{AM11}.

Let $V=\{1,\ldots,n\}$ be a finite set of \bfe{agents} and $\sgraph=(V,E,\weight)$ 
a weighted directed graph with $\weight_{ij} \in [0,1]$ being the
weight of the edge $(i,j)$. We often use the notation $i \to j$ to denote
$(i,j) \in E$ and write $i \to^*j$ if there is a path from $i$ to $j$ in the
graph $\sgraph$. 
Given a node $i$ of $G$ we denote by
$\neighbour(i)$ the set of nodes from which there is an incoming edge to $i$.
We call each $j \in \neighbour(i)$ a \oldbfe{neighbour} of $i$ in $G$.
We assume that for each node $i$ such that $\neighbour(i) \neq \ES$, $\sum_{j
\in \neighbour(i)} w_{ji} \leq 1$.
An agent $i \in V$ is said to be a
\bfe{source node} in $\sgraph$ if $\neighbour(i)=\emptyset$.

Let $\products$ be a finite set of alternatives or \bfe{products}.  By
a \bfe{social network} (from now on, just \bfe{network}) we mean a
tuple $\snet=(\sgraph,\products,\prodset,\theta)$, where $\prodset$
assigns to each agent $i$ a non-empty set of products $\prodset(i)$
from which it can make a choice. $\theta$ is a \bfe{threshold
  function} that for each $i \in V$ and $t \in \prodset(i)$ yields a
value $\theta(i,t) \in (0,1]$.
% The threshold $\theta(i,t)$ should be viewed as agent $i$'s
% resistance level to adopt a product $t$. 

Given a network $\snet$ we denote by $\srcnodes(\snet)$ the set of
source nodes in the underlying graph $\sgraph$.  One of the classes of
the networks we shall study are the ones with $\srcnodes(\snet) =
\ES$. 

\subsection{Social network games}
\label{subsec:sng}

Fix a network $\snet=(\sgraph,\products,\prodset,\theta)$.
Each agent can adopt a product from his product set or
choose not to adopt any product. We denote the latter choice by
$t_0$. 

With each network $\snet$ we associate a strategic game
$\mathcal{G}(\snet)$. The idea is that the nodes
simultaneously choose a product or abstain from choosing any.
Subsequently each node assesses his choice by comparing it with the
choices made by his neighbours.  Formally, we define the game as
follows: the players are the agents, the set of strategies for player
$i$ is $S_i :=\prodset(i) \cup \{t_0\}$,
%%\begin{itemize}
%%\item 
for $i \in V$, $t \in \prodset(i)$ and a joint strategy $\strprofile$,
let $ \inflset_i^t(\strprofile) :=\{j \in \neighbour(i) \mid s_j=t\},
$ i.e., $\inflset_i^t(\strprofile)$ is the set of neighbours of $i$
who adopted in $s$ the product $t$.  The payoff function is defined as
follows, where $\constutil$ is some positive constant given in
advance:

\begin{itemize}
\item for $i \in \srcnodes(\snet)$,
$\payoff_i(\strprofile) :=\left \{\begin{array}{ll}
                         0          & \mbox{if~~} \strprofile_i = t_0\\ 
                         \constutil & \mbox{if~~} \strprofile_i \in \prodset(i)\\
  \end{array}
  \right. $\\

\item for $i \not\in \srcnodes(\snet)$,
$\payoff_i(s) :=\left\{\begin{array}{ll}
               0 &\mbox{if~~} \strprofile_i = t_0\\
               \sum\limits_{j \in \inflset_i^t(\strprofile)} w_{ji}-\theta(i,t) & \mbox{if~~} \strprofile_i=t, \mbox{ for some } t \in \prodset(i)\\
\end{array}     
                           \right.$.\\

\end{itemize}
%%\end{itemize}

Let us explain the underlying motivations behind the above definition.
In the first entry we assume that the payoff function for the source
nodes is constant only for simplicity.  In the last section of the
paper we explain that the obtained results hold equally well in the
case when the source nodes have arbitrary positive utility for
each product.

The second entry in the payoff definition is motivated by the
following considerations.  When agent $i$ is not a source node, his
`satisfaction' from a joint strategy depends positively from the
accumulated weight (read: `influence') of his neighbours who made the
same choice as him, and negatively from his threshold level (read:
`resistance') to adopt this product.  The assumption that $\theta(i,t)
> 0$ reflects the view that there is always some resistance to adopt a
product. So when this resistance is high, it can happen that the
payoff is negative. Of course, in such a situation not adopting any product,
represented by the strategy $t_0$, is a better alternative.

The presence of this possibility allows each agent to refrain from
choosing a product.  This refers to natural situations, such as
deciding not to purchase a smartphone or not going on vacation. In the
last section we refer to an initiated research on social network games
in which the strategy $t_0$ is not present.  Such games capture
situations in which the agents have to take some decision, for
instance selecting a secondary school for their children.

By definition the payoff of each player depends only on the strategies
chosen by his neighbours, so the social network games are related to
graphical games of \cite{KLS01}. However, the underlying dependence
structure of a social network game is a directed graph and the
presence of the special strategy $t_0$ available to each player makes
these games more specific.

In what follows for $t \in \products \cup \{t_0\}$ we use the notation
$\obar{t}$ to denote the joint strategy $\strprofile$ where
$\strprofile_j=t$ for all $j \in V$.  This notation is legal only if
for all agents $i$ it holds that $t \in P(i)$.  The presence of the
strategy $t_0$ motivates the introduction and study of special types
of Nash equilibria. A Nash equilibrium $s$ is

\begin{itemize}
\item \bfe{determined} if for all $i$, $s_i \neq t_0$,

\item \bfe{non-trivial} if for some $i$, $s_i \neq t_0$, 

\item \bfe{trivial} if for all $i$, $s_i = t_0$, i.e., $s = \obar{t_0}$.
\end{itemize}

\section{Nash equilibria: general case}
\label{sec:general}

The first natural question that we address is that of the existence of
Nash equilibria in the social network games.  We establish the following
result.

\begin{theorem} \label{thm:np}
Deciding whether for a network $\snet$ the game
$\mathcal{G}(\snet)$ has a (respectively, non-trivial) Nash
equilibrium is NP-complete.
\end{theorem}

To prove it we first construct an example of a social network game
with no Nash equilibrium and then use it to determine the complexity
of the existence of Nash equilibria.

\begin{example} \label{exa:nonash}
\rm
Consider the network given in Figure~\ref{fig:noNe1}, where the product
set of each agent is marked next to the node denoting it and
the weights are labels on the edges. The source nodes are represented
by the unique product in the product set. 
\begin{figure}[ht]
\centering
$
\def\objectstyle{\scriptstyle}
\def\labelstyle{\scriptstyle}
\xymatrix@R=20pt @C=30pt{
& & \{t_1\} \ar[d]_{w_1}\\
& &1 \ar[rd]^{w_2} \ar@{}[rd]^<{\{t_1,t_2\}}\\
\{t_2\} \ar[r]_{w_1} &3 \ar[ur]^{w_2} \ar@{}[ur]^<{\{t_2,t_3\}}& &2 \ar[ll]^{w_2} \ar@{}[lu]_<{\{t_1,t_3\}} &\{t_3\} \ar[l]^{w_1}\\
%%\{t_2\} \ar[ru]_{w_1} & & & &\{t_3\} \ar[lu]^{w_1}\\
}$

\caption{\label{fig:noNe1}A network with no Nash equilibrium}
\end{figure}
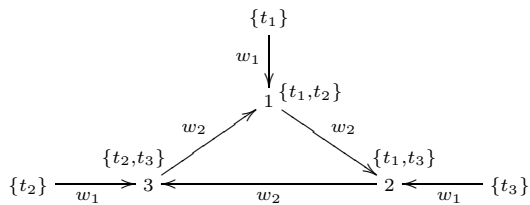

So the weights on the edges from the nodes $\{t_1\}, \{t_2\}, \{t_3\}$
are marked by $w_1$ and the weights on the edges forming the triangle
are marked by $w_2$. We assume that each threshold is a constant
$\theta$, where $ \theta < w_1 < w_2.  $ So it is more profitable to a
player residing on a triangle to adopt the product adopted by his
neighbour residing on a triangle than by the other neighbour who is a
source node. For convenience we represent each joint strategy as a
triple of strategies of players 1, 2 and 3.

It is easy to check that in the game associated with this network no
joint strategy is a Nash equilibrium.  Indeed, each agent residing on
the triangle can secure a payoff of at least $w_1 - \theta>0$, so it
suffices to analyze the joint strategies in which $t_0$ is not
used. There are in total eight such joint strategies. Here is their
listing, where in each joint strategy we underline the strategy that
is not a best response to the choice of other players:
$(\underline{t_1}, t_1, t_2)$, $(t_1, t_1, \underline{t_3})$, $(t_1,
t_3, \underline{t_2})$, $(t_1, \underline{t_3}, t_3)$, $(t_2,
\underline{t_1}, t_2)$, $(t_2, \underline{t_1}, t_3)$, $(t_2, t_3,
\underline{t_2})$, $(\underline{t_2}, t_3, t_3)$.  \HB
\end{example}

\NI
\emph{Proof of Theorem~\ref{thm:np}.}
As in \cite{AM11}, to show NP-hardness, we use a reduction from the
NP-complete PARTITION problem, which is: given $n$ positive rational
numbers $(a_1,\LL,a_n)$, is there a set $S$ such that $\sum_{i\in S}
a_i = \sum_{i\not\in S} a_i$?  Consider an instance $I$ of PARTITION.
Without loss of generality, suppose we have normalised the numbers so
that $\sum_{i=1}^n a_i = 1$. Then the problem instance sounds: does
there exist a set $S$ such that $\sum_{i\in S} a_i = \sum_{i\not\in S}
a_i = \frac{1}{2}$?
  
To construct the appropriate network we employ the networks given in
Figure~\ref{fig:noNe1} and in Figure~\ref{fig:partition}, where for
each node $i\in\{1,\LL,n\}$ we set $w_{i a} = w_{i b} = a_i$, and
assume that the threshold of the nodes $a$ and $b$ is constant and
equal $\frac12$.

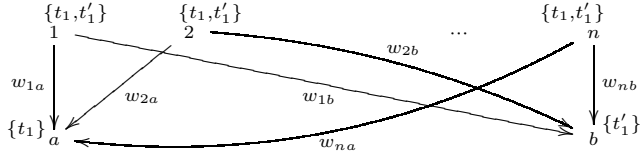
\begin{figure}[ht]
\centering
$
\def\objectstyle{\scriptstyle}
\def\labelstyle{\scriptstyle}
\xymatrix@W=10pt @R=30pt @C=35pt{
1 \ar@{}[r]^<{\{t_1,t_1'\}} \ar[d]_{w_{1a}} \ar[rrrrd]_{w_{1b}}& 2 \ar@{}[r]^<{\{t_1,t_1'\}} \ar[ld]^{w_{2a}} \ar@/^0.7pc/[rrrd]^{w_{2b}}&  &\cdots &n \ar@{}[l]_<{\{t_1,t_1'\}} \ar@/^1.5pc/[lllld]^{w_{na}} \ar[d]^{w_{nb}}\\
a \ar@{}[u]^<{\{t_1\}}  & & & &b \ar@{}[u]_<{\{t_1'\}} \\
}$
\caption{\label{fig:partition}A network related to the PARTITION problem}
\end{figure}

We use two copies of the network given in
Figure~\ref{fig:noNe1}, one unchanged and the other in which the
product $t_1$ is replaced by $t'_1$, and
construct the desired network
$\snet$ by identifying the node $a$ of the network from
Figure~\ref{fig:partition} with the node marked by $\{t_1\}$ in the
network from Figure~\ref{fig:noNe1}, and the node $b$ with the node marked
by $\{t'_1\}$ in the modified version of the network from
Figure~\ref{fig:noNe1}.

Suppose now that a solution to the considered instance of the
PARTITION problem exists, i.e., for some set $S \sse \{1, \LL, n\}$
we have $\sum_{i\in S} a_i = \sum_{i\not\in S} a_i = \frac{1}{2}$.
Consider the game $\mathcal{G}(\snet)$ and the joint strategy formed
by the following strategies:

\begin{itemize}

\item $t_1$ assigned to each node $i \in S$ in
the network from Figure~\ref{fig:partition},

\item $t'_1$ assigned to each node $i \in \{1, \LL, n\} \setminus S$
in the network from Figure~\ref{fig:partition},

\item $t_0$ assigned to the nodes $a$, $b$ and the nodes 1 in both versions of
the network from Figure~\ref{fig:noNe1},

\item $t_3$ assigned to the nodes 2, 3 in both versions
of the networks from Figure~\ref{fig:noNe1}
and the two nodes marked by $\{t_3\}$,

\item $t_2$ assigned to the nodes marked by $\{t_2\}$.

\end{itemize}

We claim that this joint strategy is a non-trivial Nash
equilibrium. Consider first the player (i.e, node) $a$. The
accumulated weight of its neighbours who chose strategy $t_1$ is
$\frac12$, so its payoff after switching to the strategy $t_1$ is 0.
Therefore $t_0$ is indeed a best response for player $a$.  For the
same reason, $t_0$ is also a best response for player $b$.  The
analysis for the other nodes is straightforward.

Conversely, suppose that a joint strategy $s$ is a Nash equilibrium
in the game $\mathcal{G}(\snet)$. Then it is also a non-trivial Nash
equilibrium. We claim that the strategy selected by the node $a$ in
$s$ is $t_0$. Otherwise, this strategy equals $t_1$ and the strategies
selected by the nodes of the network of Figure~\ref{fig:noNe1}
form a Nash equilibrium in the game associated with this network. This
yields a contradiction with our previous analysis of this
network.

So $t_0$ is a best response of the node $a$ to the strategies of the
other players chosen in $s$.  This means that $ \sum_{i \in \{1, \LL,
  n\} \mid s_i = t_1} w_{i a} \leq \frac12$. By the same reasoning
$t_0$ is a best response of the node $b$ to the strategies of the
other players chosen in $s$.  This means that $\sum_{i \in \{1, \LL,
  n\} \mid s_i = t'_1} w_{i b} \leq \frac12$.

But $\sum_{i=1}^n a_i = 1$ and for $i\in\{1,\LL,n\}$, $w_{i a} = w_{i
  b} = a_i$, and $s_i \in \{t_1, t'_1\}$.  So both above inequalities
are in fact equalities.  Consequently for $S := \{i \in \{1, \LL, n\}
\mid s_i = t_1\}$ we have $\sum_{i\in S} a_i = \sum_{i\not\in S} a_i$.
In other words, there exists a solution to the considered instance of
the PARTITION problem.

To prove that the problem lies in NP it suffices to notice that given
a network $\snet=(\sgraph,\products,\prodset,\theta)$ with $n$
nodes checking if a joint strategy is a non-trivial Nash
equilibrium can be done by means of $n \cdot |\products|$ checks, so
in polynomial time.  
%\qed

\section{Nash equilibria: special cases}
\label{sec:special}

In view of the fact that in general Nash equilibria may not exist we now consider networks
with special properties of the underlying directed graph. 
We consider first networks whose underlying graph is a directed
acyclic graph (DAG).  Intuitively, such networks correspond to
hierarchical organizations. 

\begin{theorem}
\label{thm:dag-ne}
Consider a network $\snet$ whose underlying graph is a DAG.
\begin{enumerate}[(i)]
\item $\mathcal{G}(\snet)$ always has a non-trivial Nash equilibrium.
\item Deciding whether $\mathcal{G}(\snet)$ has a determined Nash
  equilibrium is NP-complete.
\end{enumerate}
\end{theorem}

\begin{theorem}
\label{thm:cycle-NE}
Consider a network $\snet=(\sgraph,\products,\prodset,\theta)$ whose underlying graph is a simple cycle.
There is a procedure that runs in time $\mathcal{O}(|\products| \cdot n)$, where $n$
is the number of nodes in $\sgraph$, that decides whether
$\mathcal{G}(\snet)$ has a non-trivial (respectively, determined) Nash equilibrium.
\end{theorem}

\begin{theorem}
\label{thm:poa}
  The price of anarchy and the price of stability for the games associated
  with the networks whose underlying graph is a DAG or a simple cycle is
  unbounded.
\end{theorem}

Finally, we consider the case when the underlying graph $\sgraph =
(V,E)$ of a network $\snet$ has no source nodes, i.e., for all
$i \in V$, $\neighbour(i) \neq \emptyset$.  Intuitively, such a
network corresponds to a `circle of friends': everybody has a friend
(a neighbour).  For such networks we prove the following
result.

\begin{theorem} \label{thm:nosource1}
Consider a network $\snet=(\sgraph,\products,\prodset,\theta)$ whose underlying graph has no source nodes.
There is a procedure that runs in time $\mathcal{O}(|\products| \cdot n^3)$, where $n$
is the number of nodes in $\sgraph$, that decides whether
$\mathcal{G}(\snet)$ has a non-trivial Nash equilibrium.
 \end{theorem}

The proof of Theorem \ref{thm:nosource1} requires some
characterization results that are of independent interest.  The
following concept plays a crucial role.  Here and elsewhere we only
consider subgraphs that are \emph{induced} and identify each such
subgraph with its set of nodes.  (Recall that $(V',E')$ is an induced
subgraph of $(V,E)$ if $V' \sse V$ and $E' = E \cap (V' \times V')$.)

We say that a (non-empty) strongly connected subgraph (in short, SCS)
$C_t$ of $\sgraph$ is \bfe{self sustaining} for a product $t$ if for
all $i \in C_t$,

\begin{itemize}
\item $t \in \prodset(i)$,
\item $\sum\limits_{j \in \neighbour(i)\cap C_t} w_{ji} \geq
  \theta(i,t)$.
\end{itemize}

An easy observation is that if $\snet$ is a network with no
source nodes, then it always has a trivial Nash equilibrium,
$\obar{t_0}$. The following lemma
states that for such networks every non-trivial Nash equilibrium
satisfies a structural property which relates it to the set of self
sustaining SCSs in the underlying graph. We use the following
notation: for a joint strategy $\strprofile$ and product $t$,
$\agents_t(\strprofile):=\{i \in V \mid \strprofile_i=t\}$ and
$\prodset(\strprofile):=\{t \mid \exists i \in V \text{ with
}s_i=t\}$.

\begin{lemma}
\label{lm:Ne-struct}
Let $\snet=(\sgraph,\products,\prodset,\theta)$ be a network
whose underlying graph has no source nodes. If $\strprofile \neq
\obar{t_0}$ is a Nash equilibrium in $\mathcal{G}(\snet)$ then for all
products $t \in \prodset(\strprofile) \setminus \{t_0\}$ and $i \in
\agents_t(s)$ there exists a self sustaining SCS $C_t \subseteq
\agents_t(\strprofile)$ for $t$ and $j \in C_t$ such that $j \to^* i$.
\end{lemma}

\begin{lemma}
\label{lm:Ne-nosrc}
Let $\snet=(\sgraph,\products,\prodset,\theta)$ be a network
whose underlying graph has no source nodes. The joint strategy
$\obar{t_0}$ is a unique Nash equilibrium in $\mathcal{G}(\snet)$ iff
there does not exist a product $t$ and a self sustaining SCS $C_t$ for
$t$ in $\sgraph$.
\end{lemma}
\begin{proof}

\noindent($\Leftarrow$) By Lemma~\ref{lm:Ne-struct}.

%% Suppose there exists a joint strategy
%% $\strprofile \neq \obar{t_0}$ such that $\strprofile$ is a Nash
%% equilibrium. Then by Lemma~\ref{lm:Ne-struct} there exists a self
%% sustaining SCS $C_t$ for every product $t \in \prodset(\strprofile)$.

\noindent ($\Rightarrow$) Suppose there exists a self sustaining SCS
$C_t$ for a product $t$. Let $R$ be the set of nodes reachable from
$C_t$ which eventually can adopt product $t$. Formally, $R:=\bigcup_{m
  \in \nat} R_m$ where
\begin{itemize}
\item $R_0:=C_t$,
\item $R_{m+1}:=R_m \cup \{j \mid t \in \prodset(j) \mbox{ and } \sum\limits_{k \in \neighbour(j) \cap R_m} w_{kj} \geq \theta(j,t)\}$.
\end{itemize}

Let $s$ be the joint strategy such that for all $j \in R$, we have
$\strprofile_j=t$ and for all $k \in V \setminus R$, we have
$\strprofile_k=t_0$. It follows directly from the definition of $R$
that $\strprofile$ satisfies the following properties:
\begin{enumerate}
\item[(P1)] for all $i \in V$, $\strprofile_i=t_0$ or $\strprofile_i=t$,
\item[(P2)] for all $i \in V$, $\strprofile_i \neq t_0$ iff $i \in R$,
\item[(P3)] for all $i \in V$, if $i \in R$ then $\payoff_i(s) \geq 0$.
\end{enumerate}

We show that $\strprofile$ is a Nash equilibrium. Consider first any
$j$ such that $\strprofile_j=t$ (so $\strprofile_j\neq t_0$). By (P2)
$j \in R$ and by (P3) $\payoff_j(s)\geq 0$. Since
$\payoff_j(s_{-j},t_0)=0\leq \payoff_j(s)$, player $j$ does not gain
by deviating to $t_0$. Further, by (P1), for all $k \in
\neighbour(j)$, $\strprofile_k=t$ or $\strprofile_k=t_0$ and therefore
for all products $t' \neq t$ we have $\payoff_j(\strprofile_{-j},t')
<0 \leq \payoff_j(s)$. Thus player $j$ does not gain by deviating to
any product $t' \neq t$ either.

Next, consider any $j$ such that $\strprofile_j=t_0$. We have
$\payoff_j(\strprofile)=0$ and from (P2) it follows that $j \not\in
R$. By the definition of $R$ we have $\sum\limits_{k \in \neighbour(j)
  \cap R} w_{kj} < \theta(j,t)$. Thus
$\payoff_j(\strprofile_{-j},t) <0$. Moreover, for all products $t'
\neq t$ we also have $\payoff_j(\strprofile_{-j},t') <0$ for the same
reason as above. So player $j$ does not gain by a unilateral
deviation. We conclude that $\strprofile$ is a Nash equilibrium.
%\qed
\end{proof}

\NI For a product $t \in \products$, we define the set
$X_t:=\bigcap_{m \in \nat} X_t^m$, where
\begin{itemize}
\item $X_t^0:=\{i \in V \mid t \in \prodset(i)\}$,
\item $X_t^{m+1}:=\{i \in V \mid \sum_{j \in \neighbour(i) \cap
  X_t^m} w_{ji} \geq \theta(i,t)\}$.
\end{itemize}

%% We need the following characterization which leads to a direct
%% proof of the claimed result.
The following characterization leads to a direct proof of the claimed
result.

\begin{lemma}
\label{lm:proc-Ne-exists}
Let $\snet$ be a network
whose underlying graph has no source nodes. 
There exists a non-trivial Nash equilibrium in
$\mathcal{G}(\snet)$ iff there exists a product $t$ such that $X_t
\neq \emptyset$.
\end{lemma}
\begin{proof}
Suppose $\snet=(\sgraph,\products,\prodset,\theta)$.

\noindent
$(\Rightarrow)$ It follows directly from the definitions
that if there is a self sustaining SCS $C_t$ for product $t$
then $C_t \subseteq X_t$. Suppose now that for all $t$, $X_t =
\emptyset$. Then for all $t$, there is no self sustaining SCS
for $t$. So by Lemma~\ref{lm:Ne-nosrc}, $\obar{t_0}$ is a unique
Nash equilibrium.

\noindent$(\Leftarrow)$ Suppose there exists $t$ such that $X_t \neq
\emptyset$. Let $\strprofile$ be the joint strategy defined as
follows:
\[
s_i:= \begin{cases}
      t & \mathrm{if}\ i \in X_t \\
      t_0 & \mathrm{if}\ i \not\in X_t
\end{cases}
\]

By the definition of $X_t$, for all $i \in X_t$,
$\payoff_i(\strprofile) \geq 0$. So no player $i \in X_t$ gains by
deviating to $t_0$ (as then his payoff would become $0$) or to a
product $t' \neq t$ (as then his payoff would become negative since no
player adopted $t'$). Also, by the definition of $X_t$ and of the
joint strategy $\strprofile$, for all $i \not\in X_t$ and for all $t' \in
\prodset(i)$, $\payoff_i(t',\strprofile_{-i})<0$. Therefore, no player
$i \not\in X_t$ gains by deviating to a product $t'$ either. It
follows that $\strprofile$ is a Nash equilibrium.
%\qed
\end{proof}

%% This theorem leads to a direct proof of the claimed result.
%% \III

\NI \emph{Proof of Theorem~\ref{thm:nosource1}.}
%% \NI We use the following
%% procedure for checking for the existence of a non-trivial Nash
%% equilibrium. 
%% \[\mathsf{Nash(\snet)}\]
On the account of Lemma~\ref{lm:proc-Ne-exists}, the following
procedure can be used to check for the existence of a non-trivial Nash
equilibrium.  

$\mathit{found}:= \mathbf{false}$;

{\bf while} $\products \neq \emptyset$ {\bf and} $\neg \mathit{found}$ {\bf do}

\quad choose $t \in \products$;

\quad $\products := \products - \{t\}$;

\quad compute  $X_t$;

\quad $\mathit{found}:= (X_t \neq \emptyset)$

{\bf od}

{\bf return} $\mathit{found}$

\smallskip

%% On the account of Lemma~\ref{lm:proc-Ne-exists} this
%% procedure returns {\bf true} if a
%% non-trivial Nash equilibrium exists and {\bf false} otherwise.
To assess its complexity, note that for a network
$\snet=(\sgraph,\products,\prodset,\theta)$ and a fixed product $t$,
the set $X_t$ can be constructed in time $\mathcal{O}(n^3)$, where $n$
is the number of nodes in $\sgraph$. Indeed, each iteration of $X_t^m$
requires at most $\mathcal{O}(n^2)$ comparisons and the fixed point is
reached after at most $n$ steps. In the worst case, we need to compute
$X_t$ for every $t \in \products$, so the procedure runs in time
$\mathcal{O}(|\products| \cdot n^3)$. 
%\qed
\VV
In fact, the proof of
Lemma~\ref{lm:proc-Ne-exists} shows that if a non-trivial Nash
equilibrium exists, then it can be constructed in polynomial time as
well.
%%\qed

\section{The FIP and the uniform FIP}
\label{sec:uniform-FIP}
A natural question is whether the games for which we established the
existence of a Nash equilibrium belong to some well-defined class of
strategic games, for instance, games with the finite improvement
property (FIP).  
%% In view of the examples of social networks given in
%% Section \ref{sec:general}, in general these games do not have the
%% FIP. However, 
When the underlying graph of the network is a DAG, the game does
indeed have the FIP. The following theorem shows that the result can
be improved in the case of two player social network games.

\begin{theorem}
\label{thm:2-FIP}
Every two players social network game has the FIP.
\end{theorem}

\begin{proof}
By the above comment on DAGs, we can assume that the underlying graph
is a cycle, say $1 \to 2 \to 1$. Consider an improvement path $\rho$.
%% By removing if necessary some steps 
Without loss of generality we can assume that the players
alternate their moves in $\rho$.
In what follows given an element of $\rho$ (that is not the last one)
we underline the strategy
of the player who moves, i.e., selects a better response.
We call each element of $\rho$ of the type $(\underline{t}, t)$
or $(t, \underline{t})$ a \emph{match}.
Further, we shorten the statement ``each time player $i$ switches his
strategy his payoff strictly increases and it never decreases when his
opponent switches strategy'' to ``player $i$'s payoff steadily goes
up''.  

Consider two successive matches in $\rho$, based respectively on
the strategies $t$ and $t_1$. The corresponding segment of $\rho$
is one of the following four types.

\NI
\emph{Type 1}. $(\underline{t}, t) \Rightarrow^{*} (\underline{t_1}, t_1)$.
The fragment of $\rho$ that starts at 
$(\underline{t}, t)$ and finishes at $(\underline{t_1}, t_1)$ has the form:
$
(\underline{t}, t) \Rightarrow (t_2, \underline{t}) \Rightarrow^{*} (t_1, \underline{t_3}) \Rightarrow
(\underline{t_1}, t_1).
$
Then player $1$'s payoff steadily goes up.  Additionally, in the step $(t_1,
\underline{t_3}) \Rightarrow (\underline{t_1}, t_1)$ his payoff increases by
$w_{2 1}$.
In turn, in the step $(\underline{t}, t) \Rightarrow (t_2,
\underline{t})$ player $2$'s payoff decreases by $w_{1 2}$ and in the
remaining steps his payoff steadily goes up.
So $p_1(\bar{t}) + w_{2 1} < p_1(\overline{t_1})$ and $p_2(\bar{t}) - w_{1 2} < p_2(\overline{t_1})$.
%\II

\NI
\emph{Type 2}. $(\underline{t}, t) \Rightarrow^{*} (t_1, \underline{t_1})$.
Then player $1$'s payoff steadily goes up.  In turn, in the first step of
$(\underline{t}, t) \Rightarrow^{*} (t_1, \underline{t_1})$ the payoff
of player 2 decreases by $w_{1 2}$, while in the last step (in which
player 1 moves) his payoff increases by $w_{1 2}$.  So these two
payoff changes cancel against each other.
Additionally, in the
remaining steps player $2$'s payoff steadily goes up. 
So $p_1(\bar{t}) < p_1(\overline{t_1})$ and $p_2(\bar{t}) < p_2(\overline{t_1})$.
%\II

\NI
\emph{Type 3}. $(t, \underline{t}) \Rightarrow^{*} (\underline{t_1}, t_1)$.
This type is symmetric to Type 2, so
$p_1(\bar{t}) < p_1(\overline{t_1})$ and $p_2(\bar{t}) < p_2(\overline{t_1})$.
%\II

\NI
\emph{Type 4}. $(t, \underline{t}) \Rightarrow^{*} (t_1, \underline{t_1})$.
This type is symmetric to Type 1, so
$p_1(\bar{t}) - w_{2 1} < p_1(\overline{t_1})$ and $p_2(\bar{t}) + w_{1 2} < p_2(\overline{t_1})$.
%\II

Table~\ref{tab:T} summarizes the changes in the payoffs between the two matches.

\begin{table}[htbp]
\centering
%%  \begin{center}
%    \leavevmode
\begin{tabular}{|l|l|l|}
\hline 
Type & $p_1$ & $p_2$ \\
\hline 
%\hline 
1    & increases      & decreases \\
     & by $> w_{2 1}$ & by $< w_{1 2}$ \\
\hline
2, 3 & increases      & increases \\
\hline
4    & decreases      & increases \\
     & by $< w_{2 1}$ & by $> w_{1 2}$ \\
\hline
\end{tabular}
%%  \end{center}
\caption{Changes in $p_1$ and $p_2$}
\label{tab:T}
\end{table}

Consider now a match $(\underline{t}, t)$ in $\rho$ and a match $(\underline{t_1}, t_1)$
that appears later in $\rho$.
Let $T_i$ denote the number of internal segments of type $i$ that occur in the fragment of $\rho$
that starts with $(\underline{t}, t)$ and ends with $(\underline{t_1}, t_1)$.
%\II

\NI
\emph{Case 1}. $T_1 \geq T_4$.
Then Table~\ref{tab:T} shows that the aggregate increase in $p_1$ in segments of type 1
exceeds the aggregate decrease in segments of type 4.
So $p_1(\bar{t}) < p_1(\overline{t_1})$.
%\II

\NI
\emph{Case 2}. $T_1 < T_4$.
Then analogously Table~\ref{tab:T} shows that 
$p_2(\bar{t}) < p_2(\overline{t_1})$.
\II

We conclude that $t \neq t_1$. By symmetry the same conclusion holds
if the considered matches are of the form $(t, \underline{t})$ and
$(t_1, \underline{t_1})$.  This proves that each match occurs in
$\rho$ at most once.  So in some suffix $\eta$ of $\rho$ no match occurs.
But each step in $\eta$ increases the social welfare, so $\eta$ is
finite, and so is $\rho$.
%\qed
\end{proof}

The FIP ceases to hold when the underlying graph has cycles. Figure
\ref{fig:br1}(a) gives an example.  Take any threshold and weight
functions which satisfy the condition that an agent gets positive
payoff when he chooses the product picked by his unique predecessor in
the graph. Figure \ref{fig:br1}(b) then shows an infinite improvement
path. In each joint strategy, we underline the strategy that is not a
best response to the choice of other players. Note that at each step
of this improvement path a best response is used. On the other hand,
one can check that for any initial joint strategy there exists a
finite improvement path. This is an instance of a more general result
proved below.

\begin{figure}[ht]
\centering
\begin{tabular}{ccc}
$
\def\objectstyle{\scriptstyle}
\def\labelstyle{\scriptstyle}
\xymatrix@W=10pt @C=15pt{
 &1 \ar[rd]^{} \ar@{}[rd]^<{\{t_1,t_2\}} \ar@{}[rd]^>{\{t_1,t_2\}}\\
3  \ar[ru]_{} \ar@{}[ru]^<{\{t_1,t_2\}} & &2 \ar[ll]^{}\\
}$
%% \caption{\label{fig:br1} A social network with an infinite improvement path}
%% \end{figure}
&
&
%% \begin{figure}[ht]
%% \centering
$
\def\objectstyle{\scriptstyle}
\def\labelstyle{\scriptstyle}
\xymatrix@W=8pt @C=15pt @R=15pt{
(\underline{t_2},t_2,t_1)\ar@{=>}[r]& (t_1, t_2, \underline{t_1})\ar@{=>}[r]& (t_1, \underline{t_2}, t_2)\ar@{=>}[d]\\
(t_2, \underline{t_1}, t_1)\ar@{=>}[u]& (t_2, t_1, \underline{t_2})\ar@{=>}[l]& (\underline{t_1}, t_1, t_2)\ar@{=>}[l]\\
}$
%% \caption{\label{fig:br2} An infinite improvement path}
\\
\\
(a) & & (b)\\
\end{tabular}
\caption{\label{fig:br1} A social network with an infinite improvement path}
\end{figure}
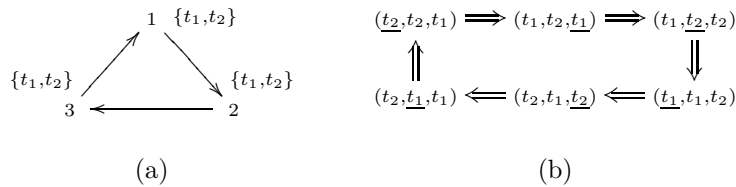 

By a \bfe{scheduler} we mean a function $f$ that given a joint
strategy $s$ that is not a Nash equilibrium selects a player who did
not select in $s$ a best response. An improvement path $\xi=
\strprofile^1,\strprofile^2,\ldots$ \bfe{conforms} to a scheduler $f$
if for all $k$ smaller than the length of $\xi$,
$\strprofile^{k+1}=(\strprofile_i',\strprofile^{k}_{-i})$, where
$f(\strprofile^k)=i$. We say that a strategic game has the
\bfe{uniform FIP} if there exists a scheduler $f$ such that all
improvement paths $\rho$ which conform to $f$ are finite.  The
property of having the uniform FIP is stronger than that of being
weakly acyclic, see \cite{AS12}.
\begin{theorem}
\label{thm:cycle-UFIP}
Let $\snet$ be a network
such that the underlying graph is a simple cycle.
Then the game $\mathcal{G}(\snet)$ has the uniform FIP.
\end{theorem}

\begin{proof}
We use the scheduler $f$ that given a joint strategy
$s$ chooses the smallest index $i$ such that
$s_i$ is not a best response to $s_{-i}$.  So this scheduler
selects a player again if he did not switch to a best response.
Therefore we can assume that each selected player immediately selects
a best response.

Consider a joint strategy $s$ taken from
a `best response' improvement path.
Observe that for all $k$ if $s_k \in P(k)$ and 
$p_k(s) \geq 0$ (so in particular if $s_k$ is a best response to $s_{-k}$), then
$s_k = s_{k \ominus 1}$. So for all $i >1$, the following property holds:
\[
\mbox{$Z(i)$: if $f(s) = i$ and $s_{i-1} \in P(i-1)$ then for all $j \in
\{n,1,\LL, i-1\}$, $s_j = s_{i-1}$.}
\]
In words: if $i$ is the first player who did not choose a best
response and player $i-1$ strategy is a product, then this product
is a strategy of every earlier player and of player $n$.
Along each `best response' improvement path that conforms to $f$ the
value of $f(s)$ strictly increases until the path terminates or at
certain stage $f(s) = n$. In the latter case if $s_{n-1} = t_0$, then
the unique best response for player $n$ is $t_0$. Otherwise $s_{n-1}
\in P(n-1)$, so on the account of property $Z(n)$ all players'
strategies equal the same product and the payoff of player $n$ is
negative (since $f(s) = n$).  So the unique best response for player
$n$ is $t_0$, as well.

This switch begins a new round with player 1 as the next scheduled player.
Player 1 also switches to $t_0$ and from now on every consecutive
player switches to $t_0$, as well. The resulting path terminates once
player $n-2$ switches to $t_0$.
%\qed
\end{proof}

\section{Concluding remarks}
\label{sec:conc}

In this paper we studied the consequences of adopting products by
agents who form a social network. To this end we analysed a natural
class of strategic games associated with the class of social networks
introduced in \cite{AM11}. 
The following table summarizes our complexity and existence results,
where we refer to the underlying graph with $n$ nodes.

\begin{center}
%\begin{table}
%\centering
\begin{tabular}{|c|c|c|c|c|}
\hline
property               & arbitrary   & DAG              & simple cycle                       & no source \\
               &             &                  &                                    &nodes \\
\hline
Arbitrary NE   & NP-complete & always exists    & always exists                      & always exists \\
Non-trivial NE & NP-complete & always exists    & $\mathcal{O}(|\products| \cdot n)$ & $\mathcal{O}(|\products| \cdot n^3)$ \\
Determined NE  & NP-complete & NP-complete      & $\mathcal{O}(|\products| \cdot n)$ & NP-complete \\
FIP            & co-NP-hard  & yes              & --                                 & co-NP-hard \\ 
Uniform FIP    & co-NP-hard  & yes              & yes                                & co-NP-hard\\
Weakly acyclic & co-NP-hard  & yes              & yes                                & co-NP-hard\\
\hline
\end{tabular}
%\caption{Summary of the complexity and existence results}
%\end{table}
\end{center}

%% \subsection{Final comments}

In the definition of the social network games we took a number of
simplifying assumptions. In particular, we stipulated that the source
nodes have a constant payoff $\constutil > 0$.  One could allow the source
nodes to have arbitrary positive utility for different products. This
would not affect any proofs. Indeed, in the Nash equilibria the source
nodes would select only the products with the highest payoff, so the
other products in their product sets could be disregarded. Further,
the FIP, the uniform FIP and weak acyclicity of a social network game
is obviously not affected by such a modification.

%% We could also allow the weights to be parametrized by a product. The
%% corresponding expression in the definition of the payoff function
%% would then become $\sum_{j \in \inflset_i^t(\strprofile)}
%% w_{ji}(t) - \theta(i,t)$. In contrast, as shown in Example~\ref{exa:12},
%% such a modification can affect some of the positive results.

The results of this paper can be slightly generalized by using a more
general notion of a threshold that would also depend on the set of
neighbours who adopted a given product. In this more general setup for
$i \in V$, $t \in \prodset(i)$ and $X \subseteq \neighbour(i)$, the
\bfe{threshold function} $\theta$ yields a value $\theta(i,t,X) \in
(0,1]$ and satisfies the following \bfe{monotonicity} condition: if
  $X_1 \subseteq X_2$ then $\theta(i,t,X_1) \geq \theta(i,t,X_2)$.
  Intuitively, agent $i$'s resistance to adopt a product decreases
  when the set of its neighbours who adopted it increases.  We decided
  not to use this definition for the sake of readability.

This work can be pursued in a couple of natural directions. One is the
study of social networks with other classes of underlying graphs.
Another is an investigation of the complexity results for other
classes of social networks, in particular for the equitable ones,
i.e., networks in which the weight functions are defined as
$\weight_{ij} = \frac{1}{|\neighbour(i)|}$ nodes $i$ and $j \in
\neighbour(i)$. One could also consider other equilibrium concepts
like the strict Nash equilibrium.

Finally, we also initiated a study of slightly different games, in
which the players are obliged to choose a product, so the games in
which the strategy $t_0$ is absent. Such games naturally correspond to
situations in which the agents always choose a product, for instance a
subscription for their mobile telephone. These games substantially
differ from the ones considered here. For example, Nash equilibrium
may not exist when the underlying graph is a simple cycle.

\bibliographystyle{abbrv}
\bibliography{/ufs/apt/bib/e}

\end{document}